\documentclass[letterpaper, 10 pt, conference]{ieeeconf}
\IEEEoverridecommandlockouts                              

\usepackage{latexsym,amssymb}
\usepackage{mathrsfs}
\usepackage{epsfig}
\usepackage{graphicx}
\usepackage{color}
\usepackage{enumerate}

\newtheorem{corollary}{Corollary}
\newtheorem{proposition}{Proposition}
\newtheorem{assumption}{Assumption}
\newtheorem{lemma}{Lemma}
\newtheorem{remark}{Remark}

\def\begequarr{\begin{eqnarray}}
\def\endequarr{\end{eqnarray}}
\def\begequarrs{\begin{eqnarray*}}
\def\endequarrs{\end{eqnarray*}}
\def\begequ{\begin{equation}}
\def\endequ{\end{equation}}
\def\begequs{\begin{equation*}}
\def\endequs{\end{equation*}}
\def\begite{\begin{itemize}}
\def\endite{\end{itemize}}

\def\begcen{\begin{center}}
\def\endcen{\end{center}}
\def\begrem{\begin{remark}\rm}
\def\endrem{\end{remark}}
\def\ba{\begin{array}}
\def\ea{\end{array}}

\def\col{ \mbox{col}\; }

\def\dst{\displaystyle}

\def\beeq#1{\begin{equation}{#1}\end{equation}}
\def\qmx#1{\begin{pmatrix}{#1}\end{pmatrix}}

\begin{document}
\title{\Huge Adaptive Semiglobal Nonlinear Output Regulation: An Extended-State Observer Approach}
\author{Lei Wang, Christopher M. Kellett
\thanks{L. Wang and C. Kellett are with Faculty of Engineering and Built Environment, University of Newcastle, Australia.
 E-mail: (wanglei\_0201@163.com; chris.kellett@newcastle.edu.au).}%
        }
\date{}
\maketitle
\begin{abstract}
This paper proposes a new extended-state observer-based framework for adaptive nonlinear regulator design of a class of nonlinear systems, in the general nonequilibrium theory. By augmenting an extended-state observer with an internal model, one is able to obtain an estimate of the term containing uncertain parameters, which then makes it possible to design an adaptive internal model in the presence of a general \emph{nonlinearly parameterized} immersion condition.
\end{abstract}

\section{Introduction}

The output regulation problem aims at controlling a disturbed system so as to achieve boundedness of the resulting trajectories and asymptotic convergence of the output towards a prescribed trajectory. Several frameworks have been established for this problem. Due to its ability to cope with uncertainties, the internal model-based method has been regarded as one of the most promising approaches, particularly since the milestone contributions for linear systems in \cite{Francis&Wonham(1976)} and nonlinear systems in \cite{IsidoriByrnes(1990)}. The main idea of this method is to appropriately incorporate the controller with the structure of an exosystem that generates the disturbance and the tracking trajectory.

In the design of an internal model-based regulator, a key step is to design an appropriate internal model to generate the steady state input such that the internal model property is fulfilled.
Several systematic design methods have been developed such as in \cite{Khalil(1994),Huang&Chen(2004),Huang2004,MarconiPraly&Isidori(2007),Lorenzo&Praly(TAC2008)}. Among them, in terms of a constructive design, significant attention has been attracted by  the ``immersion condition", which requires the solution of the regulator equations to satisfy some specific differential equations (i.e., the immersed dynamics). It is noted that, if there exist parameter uncertainties in the exosystem, the corresponding immersed dynamics would be uncertain in general, which makes the design of internal model challenging. To cope with parameter uncertainties, in \cite{Serrani&Isidori&Marconi2001} the internal model is augmented with an identifier, which is appropriately designed via the adaptive design methodology \cite{Krstic(book1995)}. Motivated by this adaptive framework, several relevant results have been reported that differ in the kind of exosystems (linear \cite{Ding(2003),Wang&Chen&Xu(TAC2018)} and nonlinear \cite{Marino&Tomei(TAC2013)} exosystems), in the kind of available information (state and output feedback), and in the kind of controlled systems (linear \cite{Marino&Tomei(TAC2003)} and nonlinear \cite{Xu&Wang&Chen(2016SCL)} systems).
On the other hand, the above mentioned ``immersion conditions" are formulated on an extra assumption that the regulator equations are solvable. This fundamentally limits the class of controlled systems that can be handled.  In \cite{Byrnes&Isidori(2003),Byrnes&Isidori(2004)}, this extra assumption is removed by taking advantage of the nonequilibrium theory of nonlinear output regulation. In \cite{Priscoli&Marconi&Isidori(2006)}, the corresponding extension to adaptive nonlinear output regulation is addressed.

Despite the aforementioned efforts, research on adaptive internal model design is still at quite an early stage. In fact, the immersion conditions in the existing design methods are quite restrictive, at least in the following two aspects. Firstly, the immersed dynamics is usually required to be linear, hence limiting the exosystem to be linear generally. It is noted that the only exception is \cite{Priscoli&Marconi&Isidori(2006)}, where the immersed dynamics is assumed to be in the output-feedback form. Besides, as in \cite{Krstic(book1995)}, the design of all adaptation laws, to the best knowledge of the authors, is  based on the idea of ``cancellation", that is to cancel the term containing the unknown parameters when computing the derivative of the Lyapunov function, which usually requires a linearly parameterized immersion condition. This in turn fundamentally limits the class of exogenous and controlled systems.

In order to deal with a broad class of exogenous and controlled systems, this paper studies the adaptive nonlinear output regulation problem with a general immersion condition, in the general nonequilibrium theory of nonlinear output regulation developed in \cite{Byrnes&Isidori(2003),Priscoli&Marconi&Isidori(2006)}. Inspired by \cite{Freidovich&Khalil(TAC(2008),Han(1995),Wang&Isidori(SCL2015)},  a new extended-state observer-based design paradigm is developed to construct an adaptive nonlinear internal model. By taking advantage of the extra state provided by the extended-state observer, one is able to obtain an estimate of the term containing the uncertain parameter to be estimated, which then can be utilised to achieve asymptotic identification.  It is noted that the proposed method allows a nonlinearly parameterized immersion condition. More specifically, the uncertain parameters in the immersed dynamics can appear in a ``monotonic-like structure", with linear parameterization as a particular case.

The paper is organized as follows. Section \ref{sec-2} gives the problem formulation and some standing assumptions. In Section \ref{sec-3}, the main results are addressed by presenting the design of the adaptive internal model and the stability analysis of the resulting closed-loop system. An illustrative example is presented in Section \ref{sec-example} to show the validity of the prosed method. A brief conclusion is made in Section \ref{sec-con}.

{\bf Notations: } For any positive integer $d$, $(A_d,B_d,C_d)$ is used to denote the matrix triplet in the prime form. Namely, $A_d$  denotes a shift matrix of dimension $d\times d$ whose all superdiagonal entries are one and other entries are all zero, $B_d$ denotes a $d\times 1$ vector whose entries are all zero except the last one which is equal to 1, and $C_d$ is a $1\times d$ vector whose entries are all zero except the first one which is equal to 1. A function $f:\mathbb{R}_+:=[0,\infty)\rightarrow\mathbb{R}_+$ is of class $\mathcal{K}$, if it is continuous, positive definite, and strictly increasing. A class $\mathcal{K}$ function is of class $\mathcal{K}_{\infty}$ if it is unbounded. A continuous function $\delta:\mathbb{R}_+\times\mathbb{R}_+\rightarrow\mathbb{R}_+$ is of class $\mathcal{KL}$ if, for each fixed $t\geq0$, the function $\delta(\cdot,t)$ is of class $\mathcal{K}$ and, for each fixed $s>0$, $\delta(s,\cdot)$ is strictly decreasing and $\lim_{t\rightarrow\infty}\delta(s,t)=0$.
\section{Preliminaries}
\label{sec-2}
\subsection{Problem Statement}
\label{subsec-problem}

Consider the system
\beeq{\label{inisys}\ba{rcl}
\dot z &=& f_0(\rho,w,z) + f_1(\rho,w,z,x)x\,\\
\dot x &=& q(\rho,w,z,x) + b(\rho,w,z,x)u\,\\
y_e &=& x
\ea}
with state $z\in\mathbb{R}^n$ and $x\in\mathbb{R}$, control input $u\in\mathbb{R}$, regulated output $y_e\in\mathbb{R}$, and in which $\rho\in\mathbb{R}^p$ and $w\in\mathbb{R}^s$ denote the exogenous input, generated by the exosystem
\beeq{\ba{rcl}\label{exosys}
\dot \rho &=& 0\,\\
\dot w &=& s(\rho,w)\,,
\ea}
with the initial conditions $\rho$ and $w_0$ taking values from compact sets $\mathcal{P}\subset\mathbb{R}^p$ and $\mathcal{W}\subset\mathbb{R}^s$, respectively.
As customary in the field of output regulation, it is assumed that $\mathcal{P}\times\mathcal{W}$ is invariant for (\ref{exosys}), and there exists a constant $b_0>0$ such that
\beeq{\label{ineq-b}
b(\rho,w,z,x) \geq b_0
}
holds for all $(\rho,w,z,x)\in\mathcal{P}\times\mathcal{W}\times\mathbb{R}^n\times\mathbb{R}$.
Additionally, functions $f_0(\cdot),f_1(\cdot),q(\cdot),b(\cdot),s(\cdot)$ are assumed to be sufficiently smooth.

In this framework, the output regulation problem of interest can be summarized as below. Given any compact sets $\mathcal{C}_z\subset\mathbb{R}^n$, $\mathcal{C}_x\subset\mathbb{R}$, all trajectories of system (\ref{inisys})-(\ref{exosys}), controlled by an output feedback regulator of the form
\beeq{\label{controller-xc}\ba{rcl}
\dot x_c &=& \varphi_c(x_c,y_e)\,,\quad x_c\in\mathbb{R}^{n_c}\\
u &=& \gamma_c(x_c,y_e)\,,
\ea}
with all initial conditions ranging over $\mathcal{P}\times\mathcal{W}\times\mathcal{C}_z\times\mathcal{C}_x\times\mathcal{C}_{x_c}$ with $\mathcal{C}_{x_c}$ being any given compact set in $\mathbb{R}^{n_c}$, are bounded and $\dst\lim_{t\rightarrow\infty}y_e(t)=0$.

With this in mind, it is observed that by viewing $x$ as the output, system (\ref{inisys}) cascaded with (\ref{exosys})  has a well-defined relative degree one, and the corresponding zero dynamics, driven by the control input
\[
u=- \frac{q(\rho,w,z,0)}{b(\rho,w,z,0)}\,,
\]
is given by
\[\ba{rcl}
\dot\rho &=& 0\,\\
\dot w &=& s(\rho,w)\,\\
\dot z &=& f_0(\rho,w,z)\,.
\ea\]
which, with $\mathbf z:=(\rho,w,z)$, can be compactly rewritten as
\beeq{\label{sys-czd}
\dot \mathbf{z} = \mathbf{f}(\mathbf{z})
}
Accordingly, we set $\mathcal{Z}:=\mathcal{P}\times\mathcal{W}\times\mathcal{C}_z$ with $\mathcal{C}_z\subset\mathbb{R}^n$ being any given compact set.

\begin{remark}
This paper is mainly interested in nonlinear systems having normal form. Although system (\ref{inisys}) has relative degree one, its extension to higher relative degree can be trivially achieved as in \cite{Isidori(1999)} by redefining a regulated output so as to reduce the relative degree to one.
\end{remark}

\subsection{Standing Assumptions}

In order to deal with a more general class of nonlinear systems, following \cite{Priscoli&Marconi&Isidori(2006)} we make some assumptions on the zero dynamics (\ref{sys-czd}).

\begin{assumption}\label{ass-MP}
There exist a nonempty, compact set $\mathcal{A}_z\subset\mathbb{R}^n$, and a class $\mathcal{K L}$ function $\delta_1(\cdot,\cdot)$  such that
for all $\mathbf{z}_0\in\mathcal{P}\times\mathcal{W}\times\mathbb{R}^n$,
\[
\mbox{dist}(\mathbf{z}(t,\mathbf{z}_0), \mathcal{Z}_c)\leq\delta_1(\mbox{dist}(\mathbf{z}_0, \mathcal{Z}_c),t)\quad \mbox{for all $t\geq0$}
\]
where $\mathcal{Z}_c:=\mathcal{P}\times\mathcal{W}\times\mathcal{A}_z$, and $\mathbf{z}(t,\mathbf{z}_0)$ denotes the solution of system (\ref{sys-czd}) passing through $\mathbf{z}_0$ at time $t=0$.

\end{assumption}

\begin{assumption}\label{ass-LES}
There exist  constants $M\geq1$, $a>0$, and $\delta_2>0$ such that for all $\mathbf{z}_0\in\mathcal{P}\times\mathcal{W}\times\mathbb{R}^n$,
\[
\mbox{dist}(\mathbf{z}_0, \mathcal{Z}_c)\leq\delta_2  \Rightarrow \mbox{dist}(\mathbf{z}(t,\mathbf{z}_0),\mathcal{Z}_c)\leq Me^{-a\,t}\mbox{dist}(\mathbf{z}_0,\mathcal{Z}_c)\,.
\]
\end{assumption}
\vspace{2mm}

\begin{remark}
Assumption \ref{ass-MP} indicates that $\mathcal{Z}_c$ is an invariant and asymptotically stable compact set under (\ref{sys-czd}). More specifically, in the sense of \cite{Byrnes&Isidori(2003)}, $\mathcal{Z}_c$ is the $\omega$-limit set of $\mathcal{P}\times\mathcal{W}\times\mathbb{R}^n$ under (\ref{sys-czd}). It can also be seen that there exists a compact set $\mathbf Z$ such that the solution of (\ref{sys-czd}) satisfies $\mathbf{z}(t,\mathbf{z}_0)\in\mathbf Z$ for all $t\geq 0$, so long as $\mathbf{z}_0\in\mathcal{Z}$.
Assumption \ref{ass-LES} implies that $\mathcal{Z}_c$ is locally exponentially stable for (\ref{sys-czd}), which plays a significant role in the subsequent analysis of asymptotic stability.
\end{remark}

\begin{remark}
Assumption \ref{ass-MP} can be regarded as the minimum-phase assumption in general nonequilibrium theory. Compared with the conventional minimum-phase assumption such as in \cite{Serrani&Isidori&Marconi2001,Huang2004}, the main benefit is that the extra assumption on the solvability of the regulator equations is removed, which broadens the class of systems that can be addressed.
\end{remark}

To this end, a general \emph{nonlinearly parameterized} immersion condition will be proposed, which leads to a constructive design of the internal model.

\begin{assumption}\label{ass-immersion}
There exist positive integers $d$ and $q$,
a $C^0$ map
\[\ba{rcl}
\theta & : &\mathcal{P} \,\rightarrow \, \mathbb{R}^q\,,\\
&& \rho \,\mapsto\, \theta(\rho)\,,
\ea\]
a $C^d$ map
\[\ba{rcl}
\tau & : &\mathcal{Z} \,\rightarrow \, \mathbb{R}^d\,,\\
&& \mathbf{z} \,\mapsto\, \tau(\mathbf{z})\,,
\ea\]
and a $C^2$ map $\phi:\mathbb{R}^p\times\mathbb{R}^d\rightarrow\mathbb{R}$ such that the following identities
\beeq{\label{tau}\ba{rcl}
\frac{\partial\tau}{\partial\mathbf{z}}\mathbf{f}(\mathbf{z}) &=& A_d\tau(\mathbf{z}) + B_d\phi(\theta(\rho),\tau(\mathbf{z}))\,\\
\mathbf{q}_0(\mathbf{z}) &=& C_d\tau(\mathbf{z})
\ea}
with $\dst\mathbf{q}_0(\mathbf{z})=-\frac{q(\rho,w,z,0)}{b(\rho,w,z,0)}$, hold for all $\mathbf{z}\in\mathcal{Z}_c$ and $\rho\in\mathcal{P}$.
\end{assumption}

\begin{remark}
In Assumption \ref{ass-immersion}, the immersed dynamics (\ref{tau}) is allowed to be dependent on the uncertain parameter $\rho$, which motivates us to incorporate the internal model with an identifier. Since $\rho$ appears only in the function $\theta(\cdot)$, for convenience we regard $\theta$ as an uncertain parameter to be estimated in the sequel, though this may result in overparameterization.

In the literature, several immersion conditions for adaptive output regulation have been proposed. It is worth noting that compared to the existing ones, Assumption \ref{ass-immersion} is much weaker, at least in the following two aspects. In previous work, the immersion map $\tau$ is required to satisfy either a linear equation (e.g. \cite{Serrani&Isidori&Marconi2001}), or a nonlinear equation but in the ``output-feedback form" (e.g. \cite{Priscoli&Marconi&Isidori(2006)}). Fundamentally, all these forms in \cite{Serrani&Isidori&Marconi2001,Priscoli&Marconi&Isidori(2006)} can be transformed to the form (\ref{tau}). Moreover, in all the previous related literature, the immersed dynamics (\ref{tau}) is required to be linearly parameterized, while this paper permits a nonlinear parameterization, with linear parameterization as a particular case.
\end{remark}

In this paper, we aim to handle a more general \emph{immersion property} having a nonlinearly parameterized function $\phi(\theta,\tau)$ in the uncertain parameter $\theta$. We require the following properties on  $\phi(\cdot,\cdot)$.

\begin{assumption}\label{ass-PE}
There exists a smooth function $\beta(\cdot):\mathbb{R}^d\rightarrow\mathbb{R}^p$ having the properties:
\begin{itemize}
  \item[(i)]  There exist $\epsilon_{0,i}>0$, $i=1,\ldots,q$ such that for any $r\in\tau(\mathcal{Z}_c)$ \footnote{For simplicity, we use  $\tau(\mathcal{Z}_c)$ to denote the set of $\tau(\mathbf{z})$ for all $\mathbf{z}\in \mathcal{Z}_c$.}, and any $s_1,s_2\in\mathcal{B}_0^q:=\{\theta\in\mathbb{R}^q:|\theta_i|\leq a_{0,i} +\epsilon_{0,i}\}$ with $a_{0,i}=\max_{\rho\in\mathcal{P}}|\theta_i(\rho)|$,  the inequality
\beeq{\label{ineq-Mon}
(s_1-\theta)^{\top}\beta(r)\frac{\partial\phi(s_2,r)}{\partial s_2}(s_1-\theta) \leq 0
}
holds, with $\theta_i$ denoting the $i$-th entry of vector $\theta$;
  \item [(ii)] For any $\mathbf{z}_0\in\mathcal{Z}_c$ and $s_1,s_2\in\mathcal{B}_0^q$,  the persistent excitation (PE) condition
\beeq{\ba{l}\label{PE-con}
\phi(s_1,\tau(\mathbf{z}(t,\mathbf{z}_0)))-\phi(s_2,\tau(\mathbf{z}(t,\mathbf{z}_0))) =0 \,\\\quad \Longrightarrow \quad s_1=s_2\,
\ea}
is fulfilled, where $\mathbf{z}(t,\mathbf{z}_0)$ denotes the trajectory of (\ref{sys-czd}) passing through $\mathbf{z}_0$ at $t=0$.
\end{itemize}
\end{assumption}

\begin{remark}
Assumption \ref{ass-PE}.(i) means that there exists a $\,$smooth function $\beta(r)$ such that for all $r\in\tau(\mathcal{Z}_c)$, the function $\beta(r)\phi(s,r)$ is \emph{monotonically decreasing} in $s\in\mathcal{B}_0^q$. In this respect, we say that the function $\phi(s,r)$ satisfying Assumption \ref{ass-PE}.(i) is in the \emph{monotonic-like structure}. If as in \cite{Serrani&Isidori&Marconi2001,Priscoli&Marconi&Isidori(2006)}, the function $\phi$ is linearly parameterized, that is $\phi(s,r)$ has the form of $s^{\top}\psi(r)$ for some function $\psi(\cdot)$, then Assumption \ref{ass-PE}.(i) can always be fulfilled by choosing $\beta(r)=\psi(r)$. Indeed, the class of functions $\phi(r,s)$ satisfying such a monotonicity condition includes not only all linearly parameterized functions, but also some nonlinearly parameterized functions, such as $\mbox{arctan}(s^\top\psi(r))$ or $\dst\frac{\psi_0(r)}{\sum_{i=1}^{p}\theta_i\psi_i(r)+\psi_{p+1}(r)}$, where the corresponding function $\beta(r)$ can be chosen as $\psi(r)$ or
$-\qmx{\psi_0(r)\psi_1(r) & \cdots & \psi_0(r)\psi_p(r)}^{\top}$, respectively.
\end{remark}

It is observed that the maps $\phi(s,r)$ and $\beta(r)$ are continuously differentiable and Assumption \ref{ass-immersion} and \ref{ass-PE} are respectively made over the compact sets $s\in\mathcal{B}_0^q$ and $(s,r)\in\mathcal{B}_0^q\times\tau(\mathcal{Z}_c)$. In view of this, there is no loss of generality to suppose that functions $\phi(\cdot,\cdot)$ and $\beta_i(\cdot)$ are globally Lipschitz and bounded, i.e., there exist $a_{1}>0$ and $a_{2,i}>0$, $i=1,\ldots,q$ such that inequalities
\beeq{\label{bound-phi-beta}
|\phi(s,r)| \leq a_{1}\,,\quad |\beta_i(r)|\leq a_{2,i}
}
with $\beta_i$ denoting the $i$-th entry of vector $\beta$, hold for all $s\in\mathbb{R}^q$, $r\in\mathbb{R}^d$.

\section{Adaptive Regulator Design}
\label{sec-3}
\subsection{Adaptive Internal Model}

With Assumption \ref{ass-immersion}, if $\theta$ were known, then  we could design an internal model of the form
\beeq{\label{IM-1}
\dot \eta = A_d\eta + B_d\phi(\theta,\eta) + v_{\eta}\,
}
in which ${\eta}\in\mathbb{R}^d$, and $v_{\eta}\in\mathbb{R}^d$ denotes the input of the internal model, and the control input can be chosen as
\beeq{\label{u}
u = v_u + C_d\eta
}
where $v_u$ is the residual input.

However, since $\theta$ is unknown, the internal model (\ref{IM-1}) is not implementable. To overcome this obstacle, an extra identifier can be used to provide an estimate of $\theta$, denoted by $\hat\theta\in\mathbb{R}^q$. It is worth noting that, due to the presence of the nonlinear parameterization, we cannot take advantage of the usual ``cancellation" idea (e.g. \cite{Serrani&Isidori&Marconi2001,Priscoli&Marconi&Isidori(2006)}).

Inspired by various important results on the design of extended-state observers (e.g. \cite{Freidovich&Khalil(TAC(2008),Han(1995),Wang&Isidori(SCL2015)}), we propose a new adaptive internal model, having the form
\beeq{\label{AIM}\ba{l}
\dot \eta = A_d\eta + B_d\phi(\hat\theta,\eta) -\mbox{satv}((A_d+\lambda I){\hat\xi})-B_d\mbox{sat}_{{d+1}}(\hat\sigma)\,\,\\
\dot {\hat\theta} = \beta(\eta)\mbox{sat}_{{d+1}}(\hat\sigma)-\mbox{dzv}(\hat\theta)\,\\
\dot {\hat\xi} = A_d\hat\xi+B_d\hat\sigma-\mbox{satv}((A_d+\lambda I){\hat\xi})-B_d\mbox{sat}_{{d+1}}(\hat\sigma) \,\\ \qquad - \Lambda_\ell  G(v_u+\hat\xi_{1})\,\\
\dot {\hat\sigma} = -\ell^{d+1} g_{d+1}(v_u+\hat\xi_{1})\,\\
\ea}
where $\hat\xi:=\col(\hat\xi_1,\ldots,\hat\xi_d)$, $\lambda>0$, $\Lambda_{\ell}=\mbox{diag}(\ell,\ldots,\ell^{d})$, $G=\mbox{col}(g_1,\ldots,g_d)$, functions $\mbox{sat}_{i}(\cdot)$ for $i=1,\ldots,d+1$ have the form
\[
\mbox{sat}_{i}(s)=\left\{\ba{l} s\,,\qquad |s|\leq l_i \\ s-\mbox{sign}(s)\dst\frac{(|s|-l_i)^2}{2} \,, \quad l_i< |s|< l_i+1\\ l_i+\frac{1}{2}\,,\quad |s|\geq l_i+1\,, \ea\right.
\]
with saturation level $l_i$,
$\mbox{satv}(\cdot):\mathbb{R}^d\rightarrow\mathbb{R}^d$ denotes a vector-valued saturation function, defined by $\mbox{satv}(s_1,\ldots,s_d)=\col(\mbox{sat}_{1}(s_1),\ldots,\mbox{sat}_{d}(s_d))$,
and $\mbox{dzv}(\cdot)$ denotes a vector-valued dead-zone function, each element of which is a function of the form
\[
\mbox{dz}_{i}(s)=\left\{\ba{l} 0 \,,\qquad |s|\leq a_{0,i} \\ c_i\dst\frac{(|s|-a_{0,i})^2}{2\epsilon_{0,i}} \mbox{sign}(s)\,,\quad a_{0,i} < |s| < a_{0,i}+\epsilon_{0,i}\,\\ c_i s-c_i\left(a_{0,i}+\dst\frac{\epsilon_{0,i}}{2}\right)\mbox{sign}(s) \,,\quad |s|\geq a_{0,i}+\epsilon_{0,i}\,. \ea\right.
\]
As it can be seen from Fig. \ref{fig0}, functions $\mbox{sat}_i$ and $\mbox{dz}_i$ are constructed to be smooth.
All design parameters $g_i$, $l_i$, and $c_i$ will be defined later in Proposition \ref{pro-1}, (\ref{sl}), and (\ref{c-i}), respectively.
\begin{figure}[thpb]
\begin{center}
\centering\includegraphics[height=30mm,width=90mm]{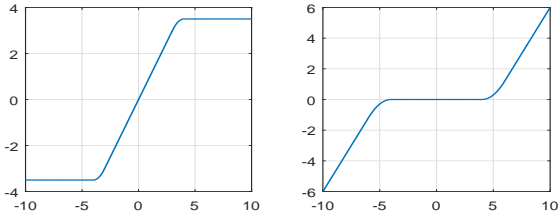} \caption{Left: plot of function $\mbox{sat}_i$ with $l_i=3$; and right: plot of function $\mbox{dz}_i$ with $c_i=1.2,a_{0,i}=4, \epsilon_{0,i}=2$.}
\label{fig0}
\end{center}
\end{figure}

By cascading system (\ref{inisys}) with the adaptive internal model (\ref{AIM}) and the control input (\ref{u}), we obtain a cascaded system of the form
\beeq{\label{sys-aug}\ba{l}
\dot\rho = 0\,\\
\dot w = s(\rho,w)\,\\
\dot z = f_0(\rho,w,z) + f_1(\rho,w,z,x)x\,\\
\dot \eta = A_d\eta + B_d\phi(\hat\theta,\eta) -\mbox{satv}((A_d+\lambda I){\hat\xi})-B_d\mbox{sat}_{{d+1}}(\hat\sigma)\,\,\\
\dot {\hat\theta} = \beta(\eta)\mbox{sat}_{{d+1}}(\hat\sigma)-\mbox{dzv}(\hat\theta)\,\\
\dot {\hat\xi} = A_d\hat\xi+B_d\hat\sigma-\mbox{satv}((A_d+\lambda I){\hat\xi})-B_d\mbox{sat}_{{d+1}}(\hat\sigma)\,\\ \qquad - \Lambda_\ell  G(v_u+\hat\xi_{1})\,\\
\dot {\hat\sigma} = -\ell^{d+1} g_{d+1}(v_u+\hat\xi_{1})\,\\
\dot x = q(\rho,w,z,x) + b(\rho,w,z,x)(C_d\eta+v_u)\,\\
\ea}

It is observed that system (\ref{sys-aug}), viewing $v_{u}$ as control input and $x$ as output, has a well-defined relative degree one, and the corresponding extended zero dynamics, forced by
\beeq{\label{f-v-x}
v_u = -C_d\eta -\frac{q(\rho,w,z,0)}{b(\rho,w,z,0)}\,,
}
can be given by
\beeq{\label{sup-augsys}\ba{l}
\dot\rho = 0\,\\
\dot w = s(\rho,w)\,\\
\dot z = f_0(\rho,w,z) \,\\
\dot \eta = A_d\eta + B_d\phi(\hat\theta,\eta) -\mbox{satv}((A_d+\lambda I){\hat\xi})-B_d\mbox{sat}_{l_{d+1}}(\hat\sigma)\,\,\\
\dot {\hat\theta} = \beta(\eta)\mbox{sat}_{{d+1}}(\hat\sigma)-\mbox{dzv}(\hat\theta)\,\\
\dot {\hat\xi} = A_d\hat\xi+B_d\hat\sigma-\mbox{satv}((A_d+\lambda I){\hat\xi})-B_d\mbox{sat}_{{d+1}}(\hat\sigma)\,\\ \dst \qquad - \Lambda_\ell  G\left(-C_d\eta -\frac{q(\rho,w,z,0)}{b(\rho,w,z,0)}+\hat\xi_{1}\right)\,\\
\dot {\hat\sigma} = -\ell^{d+1} g_{d+1}\left(-C_d\eta -\dst\frac{q(\rho,w,z,0)}{b(\rho,w,z,0)}+\hat\xi_{1}\right)\,\\
\ea}

By simple calculations, it is observed that under Assumptions \ref{ass-MP} and \ref{ass-immersion}, the adaptive controller (\ref{u})-(\ref{AIM}) fulfills the \emph{internal model property}, relative to the set $\mathcal{Z}_c$. Therefore, in light of previous analysis, according to \cite{Byrnes&Isidori(2003)}, the desired adaptive output regulation problem can be solved by the adaptive controller (\ref{u})-(\ref{AIM}) with the residual control $v_u$ having the form $v_u=-\kappa x$, if the extended zero dynamics (\ref{sup-augsys}) can be shown to to possess an asymptotically (locally exponentially) stable compact attractor.

\begin{remark}
As will be shown in next subsection, (\ref{AIM}) contains an extended state observer, i.e., the $(\hat\xi,\hat\sigma)$ dynamics, in which $\hat\sigma$ denotes the extra estimate. Using this extra estimate, we are able to take advantage of the nonlinear parameterization structure given in Assumption \ref{ass-PE}, which thus enables the identifier $\hat\theta$-dynamics to achieve an asymptotic estimate of the uncertain parameters $\theta$.
\end{remark}

\subsection{Stability Analysis of Extended Zero Dynamics (\ref{sup-augsys})}

In the previous subsection, with the design of (\ref{AIM}) for system (\ref{sys-aug}), we obtain an extended zero dynamics (\ref{sup-augsys}), whose stability analysis will be presented in the sequel.

As before, we write $\mathbf{z}=(\rho,w,z)$.
Consider the change of coordinates $\tilde\eta=\eta-\tau(\mathbf{z})$.
This, recalling (\ref{sys-czd}), transforms (\ref{sup-augsys}) to the form
\beeq{\label{sup-augsys-2-1}\ba{l}
\dot \mathbf{z} = \mathbf{f}(\mathbf{z})\,\\
\dot {\tilde\eta} = A_d\tilde\eta + B_d[\phi(\hat\theta,\tilde\eta+\tau)-\phi(\theta,\tau)]\,\\ \qquad -\mbox{satv}((A_d+\lambda I){\hat\xi})-B_d\mbox{sat}_{{d+1}}(\hat\sigma) + \varsigma(\mathbf{z})\,\\
\dot {\hat\theta} = \beta(\tilde\eta+\tau)\mbox{sat}_{{d+1}}(\hat\sigma)-\mbox{dzv}(\hat\theta)\,\\
\dot {\hat\xi} = A_d\hat\xi + B_d\hat\sigma-\mbox{satv}((A_d+\lambda I){\hat\xi}) -B_d\mbox{sat}_{{d+1}}(\hat\sigma)\,\\ \,\quad\,+ \Lambda_\ell  G(\tilde\eta_1-\hat\xi_{1})\,\\
\dot {\hat\sigma} = \ell^{d+1} g_{d+1}(\tilde\eta_1-\hat\xi_{1})\,\\
\ea}
where
\[\ba{l}
\dst\varsigma(\mathbf{z}) = A_d\tau(\mathbf{z}) + B_d\phi(\theta,\tau(\mathbf{z})) -\frac{\partial\tau(\mathbf{z})}{\partial w}s(w)\,\\ \qquad -\dst\frac{\partial\tau(\mathbf{z})}{\partial z}f_0(\mathbf{z})
\ea\]
is a term which vanishes in $\mathcal{Z}_c$ by Assumption \ref{ass-immersion}.

Let $\varsigma_i(\mathbf{z})$ denote the $i$-th element of the vector $\varsigma(\mathbf{z})$, and then set $\xi:=\col(\xi_1,\ldots,\xi_d)$ with
\[\ba{rcl}
\xi_1 &=& \tilde\eta_1\,\\
\dst \xi_2 &=& \tilde\eta_2 + \varsigma_1(\mathbf{z})\,\\
\dst \xi_i &=& \tilde\eta_i + \sum_{j=1}^{i-2}L^{i-j-1}_{\mathbf{f}}\varsigma_{j+1}(\mathbf{z})+\varsigma_{i-1}(\mathbf{z})\,,\quad 3\leq i\leq d
\ea\]
with $L$ denoting the Lie derivative, which also suggests that $\tilde\eta = \xi - \bar\varsigma(\mathbf{z})$ for an appropriately defined function $\bar\varsigma(\mathbf{z})$, satisfying $\bar\varsigma(\mathbf{z})=0$ for all $\mathbf{z}\in\mathcal{Z}_c$.

In view of the previous analysis, (\ref{sup-augsys-2-1}) can be rewritten as
\beeq{\label{sup-augsys-2}\ba{l}
\dot \mathbf{z} = \mathbf{f}(\mathbf{z})\,\\
\dot {\xi} = A_d\xi + B_d[\phi(\hat\theta,\xi+\tau(\mathbf{z})-\bar\varsigma(\mathbf{z}))-\phi(\theta,\tau(\mathbf{z}))]\,\\ \qquad -\mbox{satv}((A_d+\lambda I){\hat\xi})-B_d\mbox{sat}_{{d+1}}(\hat\sigma) + B_d\nu(\mathbf{z})\,\\
\dot {\hat\theta} = \beta(\xi+\tau(\mathbf{z}) - \bar\varsigma(\mathbf{z}))\mbox{sat}_{{d+1}}(\hat\sigma) - \mbox{dzv}(\hat\theta)\,\\
\dot {\hat\xi} = A_d\hat\xi +B_d\hat\sigma-\mbox{satv}((A_d+\lambda I){\hat\xi}) -B_d\mbox{sat}_{{d+1}}(\hat\sigma)\,\\ \,\quad\,+ \Lambda_\ell  G(\xi_1-\hat\xi_{1})\,\\
\dot {\hat\sigma} = \ell^{d+1} g_{d+1}(\xi_1-\hat\xi_{1})\,\\
\ea}
where $\dst \nu(\mathbf{z})= \sum_{i=1}^{d-1} L^{d-i}_{\mathbf{f}}\varsigma_i(\mathbf{z})+\varsigma_d(\mathbf{z})$.
It is noted that $\nu(\mathbf{z}) = 0$ for all $\mathbf{z}\in\mathcal{Z}_c$ and there exists a constant $a_3>0$ such that for all $\mathbf{z}\in\mathbf{Z}$,
\beeq{\label{bound-nu}
|\nu(\mathbf{z})|\leq a_3\,.
}

It then can be seen that the $(\hat\xi,\hat\sigma)$ dynamics in (\ref{sup-augsys-2}) can be viewed as an extended-state observer of the $\xi$ dynamics, with observer states $\hat\xi$ and $\hat\sigma$ respectively being used to estimate the variables $\xi$, and the ``perturbation" term
\beeq{\label{sigma}
\sigma:=\phi(\hat\theta,\xi+\tau(\mathbf{z})-\bar\varsigma(\mathbf{z}))-\phi(\theta,\tau(\mathbf{z}))+\nu(\mathbf{z})\,.
}
This observation thus motivates us to analyse the asymptotic stability of the extended zero dynamics (\ref{sup-augsys-2}) by using the nonlinear separation principle \cite{Isidori(1999)}, but in the general nonequilibirum theory.

Fix all coefficients of the dead-zone function $\mbox{dzv}(\cdot)$ as
\beeq{\label{c-i}
c_i> \frac{4a_1a_{2,i}+2a_{2,i}a_3}{\epsilon_{0,i}}\,,\quad i=1,\ldots,d\,,
}
with constants $a_1$, $a_{2,i}$, $a_3$, and $\epsilon_{0,i}$ being given by (\ref{bound-phi-beta}), (\ref{bound-nu}), and Assumption \ref{ass-PE}.(i).

With the above choice of $c_i$'s in mind, to apply the nonlinear separation principle to analyze the asymptotic stability of system (\ref{sup-augsys-2}), it is natural to first consider the \emph{auxiliary system}
\beeq{\label{auxisys}\ba{l}
\dot \mathbf{z} = \mathbf{f}(\mathbf{z})\,\\
\dot {\xi} = -\lambda\xi\,\\
\dot {\hat\theta} = \beta(\xi+\tau(\mathbf{z}) - \bar\varsigma(\mathbf{z}))[\phi(\hat\theta,\xi+\tau(\mathbf{z})-\bar\varsigma(\mathbf{z}))-\phi(\theta,\tau(\mathbf{z}))] \,\\ \qquad +\beta(\xi+\tau(\mathbf{z})- \bar\varsigma(\mathbf{z}))\nu(\mathbf{z}) - \mbox{dzv}(\hat\theta)\,\\
\ea}
whose stability properties can be characterized as below.

\begin{lemma}\label{lemma--1}
Suppose Assumptions \ref{ass-MP}, \ref{ass-immersion}, and \ref{ass-PE} hold.
Then the set $\mathcal{A}_a:=\mathcal{Z}_c\times\{0\}\times\{\theta\}$ is asymptotically stable under the flow (\ref{auxisys}), for every initial condition $(\mathbf{z}_0,\xi_0,\hat\theta_0)$ ranging over the set $\mathcal{M}:=\mathcal{Z}\times\mathbb{R}^d\times\mathbb{R}^p$.
\end{lemma}
\begin{proof}
The proof is given in Appendix \ref{subsec-proof-1}.
\end{proof}

\begin{lemma}\label{lemma-2}
Suppose Assumptions \ref{ass-LES}, \ref{ass-immersion}, and \ref{ass-PE} hold. Then the set $\mathcal{A}_a$ under the flow (\ref{auxisys}) is locally exponentially stable.
\end{lemma}
\begin{proof}
The proof is given in Appendix \ref{subsec-proof-2}.
\end{proof}

By setting $\tilde\theta=\hat\theta-\theta$ and letting $\mathbf{z}(t)$ denote the solution of system $\dot \mathbf{z} = \mathbf{f}(\mathbf{z})$ with initial condition ranging over $\mathcal{Z}$, system (\ref{auxisys}) can be rewritten as a \emph{nonautonomous system}
\beeq{\label{auxisys-na}\ba{l}
\dot {\xi} = -\lambda\xi\,\\
\dot {\tilde\theta} = \beta(\xi+\tau(\mathbf{z}(t)) - \bar\varsigma(\mathbf{z}(t)))\cdot \,\\ \cdot\left[\phi(\tilde\theta+\theta,\xi+\tau(\mathbf{z}(t))-\bar\varsigma(\mathbf{z}(t)))-\phi(\theta,\tau(\mathbf{z}(t)))\right] \,\\ +\beta(\xi+\tau(\mathbf{z}(t))- \bar\varsigma(\mathbf{z}(t)))\nu(\mathbf{z}(t)) - \mbox{dzv}(\tilde\theta+\theta)\,.\\
\ea}

With Lemma \ref{lemma--1}, and recalling Assumption \ref{ass-MP} and the fact that $\mathbf{z}(t)$ are captured by the compact set $\mathbf{Z}$, we can conclude the following result, whose proof can be obtained by simply adapting the proof of \cite[Theorem 3.1]{Serrani&Isidori&Marconi2001} to the present framework and is thus omitted.
\begin{corollary}
Suppose Assumptions \ref{ass-MP}, \ref{ass-immersion}, and \ref{ass-PE} hold. The zero equilibrium of nonautonomous system (\ref{auxisys-na}) is uniformly asymptotically stable, for all $\mathbf{z}_0\in\mathcal{Z}$.
\end{corollary}

By letting $\mathbf{x}_a = \col(\xi,\tilde\theta)$, system (\ref{auxisys-na}) can be compactly rewritten as
\beeq{
\dot \mathbf{x}_a = \mathbf{f}_a(\mathbf{z}(t),\mathbf{x}_a)
}
where $\mathbf{f}_a(\mathbf{z}(t),\mathbf{x}_a)$ is continuously differentiable. It is worth noting that by constructing functions $\beta(\cdot)$ and $\phi(\cdot,\cdot)$ to be globally bounded and Lipschitz, and since $\mathbf{z}(t)\in\mathbf{Z}$ for all $t\geq0$, there exists a $\varpi_0>0$ such that
\[
\left\|\frac{\partial\mathbf{f}_a(\mathbf{z}(t),\mathbf{x}_a)}{\partial \mathbf{x}_a}\right\| \leq \varpi_0\,.
\]
According to \cite[Theorem 4.16]{Khalil(2002)}, this property, together with Lemma \ref{lemma--1}, indicates that there exist a smooth, positive definite function $W_a(t,\mathbf{x}_a)$, and class $\mathcal{K}_{\infty}$ functions $\alpha_1,\alpha_2$, $\alpha_3$, and $\alpha_4$ such that
\beeq{\label{W-a}\ba{c}
\alpha_1(|\mathbf{x}_a|) \leq W_a(t,\mathbf{x}_a) \leq \alpha_2(|\mathbf{x}_a|)\,\\
\dst\frac{\partial W_a}{\partial t} + \frac{\partial W_a}{\partial \mathbf{x}_a}\dot \mathbf{x}_a\leq -\alpha_3(|\mathbf{x}_a|)\,\\
\dst\left|\frac{\partial W_a}{\partial \mathbf{x}_a}\right| \leq \alpha_4(|\mathbf{x}_a|)\,.
\ea}

With this in mind, we turn to system (\ref{sup-augsys-2}) and  define the rescaled estimate errors as
\beeq{\label{e}
\tilde{\xi} = \ell^{d+1}\Lambda_{\ell}^{-1}(\xi-\hat\xi)\,,\quad
\tilde\sigma = \sigma - \hat\sigma\,.
}
Taking time derivatives of these errors along (\ref{sup-augsys-2}) yields
\beeq{\ba{rcl}
\dot{\tilde\xi} &=& \ell(A_d-G\,C_d)\tilde\xi + \ell\,B_d\tilde\sigma\,\\
\ea}
and
\beeq{\ba{l}
\dot{\tilde\sigma} = -\ell g_{d+1}\tilde\xi_1 + \dot\phi(\hat\theta,\xi+\tau(\mathbf{z})-\bar\varsigma(\mathbf{z}))-\dot\phi(\theta,\tau(\mathbf{z}))\,\\
 \quad = -\ell g_{d+1}\tilde\xi_1 + \Delta_e\,
\ea}
where the term $\Delta_e$ is defined by
\beeq{\label{Delta-e}\ba{l}
\Delta_e = \dst\frac{\partial\phi(\hat\theta,\xi+\tau)}{\partial\hat\theta}\dot {\hat\theta}+ \frac{\partial\phi(\hat\theta,\xi+\tau)}{\partial\xi}\dot\xi  \\ \qquad \dst +  \left[\frac{\partial\phi(\hat\theta,\xi+\tau-\bar\varsigma)}{\partial\tau} - \frac{\partial\phi(\theta,\tau)}{\partial\tau}\right]\dot\tau(\mathbf{z}) \\ \qquad \dst - \frac{\partial\phi(\hat\theta,\xi+\tau-\bar\varsigma)}{\partial\bar\varsigma} \dot{\bar\varsigma}(\mathbf{z})\,.
\ea}
It is worth noting that $\Delta_e=0$ for all $(\mathbf{z},\mathbf{x}_a)\in\mathcal{A}_a$ and $e=0$, and due to the presence of saturation functions, $|\Delta_e|$ is bounded for all bounded $(\mathbf{z},\mathbf{x}_a)$, uniformly in $(\tilde\xi,\tilde\sigma)$.

Putting these equations together and letting $e=\col(\tilde\xi,\tilde\sigma)$, we can compactly obtain
\beeq{
\dot e = \ell F_e e + B_{d+1}\Delta_e
}
where $F_e$ is  defined by
\[\ba{l}
F_e = \qmx{-G & I_d  \cr  -g_{d+1} & 0}\,\\
\ea\]
This allows us to rewrite (\ref{sup-augsys-2}) as
\beeq{\label{compact-zd}\ba{rcl}
\dot \mathbf{z} &=& \mathbf{f}(\mathbf{z})\,\\
\dot \mathbf{x}_a &=& \mathbf{f}_a(\mathbf{z},\mathbf{x}_a) + \Xi(\mathbf{z}(t),\mathbf{x}_a,e)\,\\
\dot e &=& \ell F_e e + B_{d+1}\Delta_e\,.
\ea}

Thus, given any compact set $\mathcal{C}_{\mathbf{x}}\in\mathbb{R}^{p+d}$, choose $c$ such that $\mathcal{A}_c\supset\mathcal{C}_{\mathbf{x}}$ with
\[
\mathcal{A}_c=\{\mathbf{x}_a: \alpha_1(|\mathbf{x}_a|)\leq c\}\,,
\]
and let
\[
\Omega_{c+1}=\{\mathbf{x}_a: \alpha_1(|\mathbf{x}_a|)\leq \max_{\mathbf{x}_a\in\mathcal{A}_c}\alpha_2(|\mathbf{x}_a|)+1\}\,.
\]
It is clear that $\mathcal{A}_c\subset\Omega_{c+1}$.
Then, choose the saturation levels as
\beeq{\label{sl}\ba{l}
\dst l_i = \max_{\mathbf{x}_a\in\Omega_{c+1}}|\lambda \xi_i+\xi_{i+1}| + 1\,,\quad 1\leq i\leq d-1\,\\
\dst l_{d} = \max_{\mathbf{x}_a\in\Omega_{c+1}}|\lambda \xi_d| + 1\,\\
\dst l_{d+1} = \max_{(\mathbf{z},\mathbf{x}_a)\in\mathbf{Z}\times\Omega_{c+1}}\left|\phi(\hat\theta,\varphi_\eta(\xi+\tau(\mathbf{z})- \bar\varsigma(\mathbf{z}(t)))) \right.\, \\ \qquad \qquad \left.-\phi(\theta,\tau(\mathbf{z}))\right|+1\,.
\ea}
With the above choice of $l_i$'s, it can be observed that for all $(\mathbf{z},\mathbf{x}_a)\in\mathbf{Z}\times\Omega_{c+1}$, $\Xi(\mathbf{z},\mathbf{x}_a,e)$ is bounded  uniformly in $e$, and $\Xi(\mathbf{z},\mathbf{x}_a,0)=0$.

Therefore, from the standard arguments of nonlinear separation principles \cite{Isidori(1999)}, semiglobal asymptotic stability of the closed-loop system (\ref{sup-augsys-2}) can be easily concluded as below.
\begin{proposition}\label{pro-1}
Suppose Assumptions \ref{ass-MP}, \ref{ass-LES} and \ref{ass-PE} hold. Given any compact sets $\mathcal{C}_{\mathbf{x}}\in\mathbb{R}^{q+d}$ and $\mathcal{C}_e\in\mathbb{R}^{d+1}$,
and choosing $g_i$'s such that matrix $F_e$ is Hurwitz, there exists $\ell^\ast>1$ such that for all $\ell\geq \ell^\ast$ the set
\[
\{(\mathbf{z},\hat\theta,\xi,\hat\xi,\hat\sigma):\mathbf{z}\in\mathcal{Z}_c,\xi=0,\hat\theta=\theta,\hat\xi=0,\hat\sigma=0\}\,
\]
under the flow (\ref{sup-augsys-2}) is locally exponentially stable, and asymptotically stable for all initial conditions in $\mathcal{Z}\times \mathcal{C}_{\mathbf{x}}\times\mathcal{C}_e$.
\end{proposition}

\subsection{Adaptive Output Regulation}

We now turn to the extended system (\ref{sys-aug}).
As mentioned before, this system, viewed as a system with input $v_u$ and output $y_e=x$, has relative degree one.
By taking the change of variables
\[\ba{rcl}
\dst\check{\xi}&:=&\hat\xi+\Lambda_{\ell}G \dst\int_0^{x}\frac{1}{b(\rho,w,z,s)}ds\,\\
\dst\check{\sigma}&:=&\hat\sigma+ \ell^{d+1}g_{d+1}\dst\int_0^{x}\frac{1}{b(\rho,w,z,s)}ds\,\\
\ea\]
system (\ref{sys-aug}) can be rewritten  in ``normal form" as
\beeq{\label{supsys-2}\ba{l}
\dot\rho = 0\,\\
\dot w = s(\rho,w)\,\\
\dot z = f_0(\rho,w,z) + f_1(\rho,w,z,x)x\,\\
\dst\dot \eta = A_d\eta + B_d\phi(\hat\theta,\eta)-\mbox{satv}((A_d+\lambda I){\check{\xi}})\,\\ \dst
    \,\quad\, -B_d\mbox{sat}_{{d+1}}(\hat\sigma) + \mu_1(\rho,w,z,\hat\theta,\eta,\check{\xi},\check{\sigma},x)x\,\\
\dot {\hat\theta} = \beta(\eta)\mbox{sat}_{{d+1}}(\check{\sigma}) + \mu_2(\rho,w,z,\hat\theta,\eta,\check{\xi},\check{\sigma},x)x\,\\
\dst\dot {\check{\xi}} = A_d\check{\xi} +B_d\check{\sigma}-\mbox{satv}((A_d+\lambda I){\check{\xi}}) \dst-B_d\mbox{sat}_{{d+1}}(\check{\sigma})\,\\ \,\dst + \Lambda_\ell  G\left(C_d\eta+\frac{q(\rho,w,z,0)}{b(\rho,w,z,0)}-\hat\xi_{1}\right) + \mu_3(\rho,w,z,\hat\theta,\eta,\check{\xi},\check{\sigma},x)x\,\\
\dst\dot {\check{\sigma}} = \ell^{d+1} g_{d+1}\left(C_d\eta+\frac{q(\rho,w,z,0)}{b(\rho,w,z,0)}-\check{\xi}_{1}\right)\,\\ \qquad + \mu_4(\rho,w,z,\hat\theta,\eta,\check{\xi},\check{\sigma},x)x\,\\
\dst\dot x = q(\rho,w,z,x)- b(\rho,w,z,x)\frac{q(\rho,w,z,0)}{b(\rho,w,z,0)} \\ \dst\qquad + b(\rho,w,z,x)\left(C_d\eta+\frac{q(\rho,w,z,0)}{b(\rho,w,z,0)}\right)+ b(\rho,w,z,x)v_u\,\\
\ea}
in which $\mu_i(\cdot)$, $i=1,\ldots,4$ are continuous functions.

Bearing in mind the results in Proposition \ref{pro-1} and recalling \cite[Proposition 4]{Priscoli&Marconi&Isidori(2006)}, we can choose $v_u$ for system (\ref{supsys-2}) as the form
\beeq{\label{v-u}
v_u=-\kappa x\,,
}
and the following conclusion can be easily made.
\begin{proposition}
Consider system (\ref{inisys}) with exosystem (\ref{exosys}) and controller (\ref{controller-xc}) having the form (\ref{AIM}) and (\ref{v-u}). Suppose Assumptions \ref{ass-MP}-\ref{ass-PE} hold. Given any compact sets $\mathcal{C}_z\subset\mathbb{R}^n$, $\mathcal{C}_x\subset\mathbb{R}$ and $\mathcal{C}_{x_c}\subset\mathbb{R}^{2d+q+1}$, and choosing $g_i$'s such that matrix $F_e$ is Hurwitz, there exist $\ell^\ast>1$ and a positive function $\kappa^\ast(\cdot)$ such that for all $\ell>\ell^\ast$ and $\kappa\geq\kappa^\ast(\ell)$, the resulting trajectories of the closed-loop system are bounded and $x(t)\rightarrow 0$ as $t\rightarrow \infty$, with the domain of attraction that contains $\mathcal{C}_z\times\mathcal{C}_x\times\mathcal{C}_{x_c}$.
\end{proposition}

\section{An illustrative Example}
\label{sec-example}

Consider the output regulation problem for the nonlinear system
\beeq{\label{example-sys}\ba{rcl}
\dot \zeta_1 &=& \rho\zeta_1 - (\zeta_1+w_1)^3+w_2+\zeta_2\,\\
\dot \zeta_2 &=& \zeta_3\,\\
\dot \zeta_3 &=& -w_1 + \zeta_1\zeta_2 + u\,\\
y_e &=& \zeta_1
\ea}
in which $(\zeta_2,\zeta_3)$ are measurable states, and  the exogenous variables $w_1,w_2$ are generated by an uncertain nonlinear oscillator
\beeq{\label{example-exosys}\ba{rcl}
\dot w_1 &=& w_2\,\\
\dot w_2 &=& -w_1 + (1-w_1^2)\dst\frac{w_2}{1+\rho w_1}
\ea}
where $\rho$ is a constant unknown parameter satisfying $\rho \in[-0.2,0.2]$. The trajectories of (\ref{example-exosys}) at each $\rho\in\{-0.2,0,0.2\}$ are given in Fig.\ref{fig1}. It can be seen that for any $\rho \in[-0.2,0.2]$ there exists a limit cycle, that is an invariant set $\mathcal{W}$ for (\ref{example-exosys}), and particularly $\mathcal{W}\subset\{(w_1,w_2): |w_i|\leq 3, i=1,2\}$.

\begin{figure}[thpb]
\begin{center}
\centering\includegraphics[height=35mm,width=80mm]{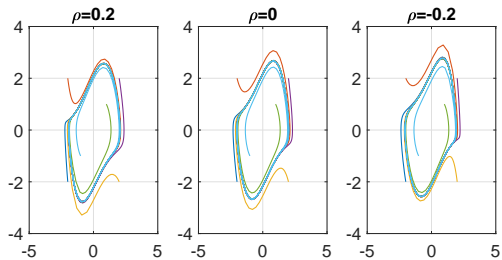} \caption{Phase portrait of (\ref{example-exosys}) at $\rho=0.2$,  $\rho=0$ and  $\rho=-0.2$.}
\label{fig1}
\end{center}
\end{figure}

Note that, when $w_1=w_2\equiv0$, system (\ref{example-sys}), regarded as a system with input $u$ and output $y_e$, has relative degree $2$ and a zero dynamics as $\dot \zeta_1 = \rho\zeta_1 - \zeta_1^3$, whose zero equilibrium point is unstable when $\rho>0$ and stable when $\rho\leq0$.  Thus, the conventional methods \cite{Huang2004,Serrani&Isidori&Marconi2001} based on equilibrium theory cannot be applied.

Following the design paradigm proposed in this paper, we first set $z_1=\zeta_1$, $z_2=\zeta_2$ and $x=\zeta_2+\zeta_3$, which reduces the relative degree of system (\ref{example-sys}) to one, leading to the form
\beeq{\label{example-sys-2}\ba{rcl}
\dot z_1 &=& \rho z_1 - (z_1+w_1)^3+w_2+z_2\,\\
\dot z_2 &=& -z_2 + x\,\\
\dot x &=& -w_1 -z_2 + z_1z_2+ x + u\,.\\
\ea}
The zero dynamics of system (\ref{example-exosys})-(\ref{example-sys-2}) with respect to input $u$ and output $x$, forced by the control input $u=w_1+z_2-z_1z_2$, can be described as
\[
\ba{rcl}
\dot \rho  &=& 0\\
\dot w_1 &=& w_2\,\\
\dot w_2 &=& -w_1 + (1-w_1^2)\dst\frac{w_2}{1+\rho w_1}\\
\dot z_1 &=& \rho z_1 - (z_1+w_1)^3+w_2+z_2\,\\
\dot z_2 &=& -z_2\,.
\ea
\]
Then, by some simple calculations, it can be seen that Assumptions \ref{ass-MP} and \ref{ass-LES} are fulfilled for some $\omega$-limit set on which $z_2=0$. In view of this, we proceed to verify Assumptions \ref{ass-immersion} and \ref{ass-PE}. Observe that in the present setting,  Assumption \ref{ass-immersion} is fulfilled with the map $\tau:=(\tau_1,\tau_2)=(w_1,w_2)$ satisfying the equations
\[\ba{rcl}
\dot\tau_1 = \tau_2\,,\qquad
\dot\tau_2 = \phi(\theta,\tau) \,\\
\ea\]
where $\theta=\rho$ and function $\phi(\theta,\tau)=-\varphi_s(\tau_1)+(1-\varphi_s^2(\tau_1))\dst\frac{\varphi_s(\tau_2)}{1+\varphi_s(\theta)\varphi_s(\tau_1)}$ with
\[\ba{l}
\varphi_s(\tau_i) = \tau_i\,,\quad \mbox{for } |\tau_i|\leq 3\,\\
\varphi_s(\theta) = \theta\,,\quad \mbox{for } |\theta|\leq 0.2\,.
\ea\]
Moreover, by choosing $\beta(\tau)=(1-\varphi_s^2(\tau_1))\varphi_s(\tau_1)\varphi_s(\tau_2)$, it can be easily found that the function $\beta(\tau)\phi(\theta,\tau)$ is strictly decreasing in $|\theta|\leq 0.25$, for all $\tau\in\mathcal{W}$. In this way, Assumption \ref{ass-PE} is also fulfilled.

Therefore, the adaptive internal model-based regulator (\ref{AIM}) and (\ref{u}) can be employed to handle the nonlinear output regulation problem at hand. Figure \ref{fig2} shows simulation results for $\rho=0.2$, and the design parameters $\ell=10$ and $\kappa=30$. It demonstrates that the regulated output $y_e$ converges to zero asymptotically and the parameter estimate $\hat\theta$ converges to the real value.

\begin{figure}[thpb]
\begin{center}
\centering\includegraphics[height=35mm,width=50mm]{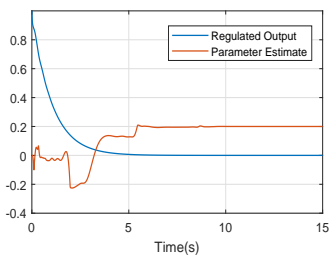} \caption{Trajectories of regulated output $y_e(t)$ and parameter estimate $\hat\theta(t)$.}
\label{fig2}
\end{center}
\end{figure}

\section{Conclusion}
\label{sec-con}

This paper studies the adaptive output regulation problem for a class of nonlinear systems using the general nonequilibrium theory developed in \cite{Byrnes&Isidori(2003)}.  By incorporating an extended-state observer into the  adaptive internal model, a new approach is proposed to deal with adaptive nonlinear regulation, which allows for more general nonlinearly parameterized immersion conditions.

\appendix

\section{Proof}
\subsection{Proof of Lemma 1}
\label{subsec-proof-1}

The proof mainly follows the nonequilibrium theory developed in \cite{Byrnes&Isidori(2003)}. First, we will show that the trajectories of system (\ref{auxisys}) are bounded, i.e. there is no finite-time escape. By Assumption \ref{ass-MP} and the choice of $\lambda>0$, it can be easily seen that $\mathbf{z}(t)$ and $\xi(t)$ are bounded. To show $\hat\theta(t)$ is also bounded, we let $\hat\theta_i$ denote the $i$-th element of vector $\hat\theta$ and choose $V_{\hat\theta,i}={1\over2}|\hat\theta_i|^2$, $i=1,\ldots,p$. Taking the time derivative of $V_{\hat\theta,i}$ along the bottom equation of (\ref{auxisys}) yields that
\[\ba{rcl}
\dot V_{\hat\theta,i} &=& \hat\theta_i\beta_i[\phi(\hat\theta,\xi+\tau(\mathbf{z})-\bar\varsigma(\mathbf{z}))-\phi(\theta,\tau(\mathbf{z}))] \,\\ &&\qquad +\hat\theta_i\beta_i\nu(\mathbf{z}) - \hat\theta_i\mbox{dz}_{i}(\hat\theta_i)\,\\
&\leq& -\hat\theta_i\mbox{dz}_{i}(\hat\theta_i) + (2a_1a_{2,i}+a_{2,i}a_3)|\hat\theta_i|
\ea\]
where (\ref{bound-phi-beta}) and (\ref{bound-nu}) are used to get the inequality.

If $|\hat\theta_i|\geq a_{0,i}+\epsilon_{0,i}$, then
\[\ba{rcl}
\dot V_{\hat\theta,i} &\leq& -c_i\hat\theta_i[\hat\theta_i- (a_{0,i}+\frac{\epsilon_{0,i}}{2})] + (2a_1a_{2,i}+a_{2,i}a_3)|\hat\theta_i|\,\\
&\leq& -\frac{\epsilon_{0,i}}{2}(c_i - \frac{4a_1a_{2,i}+2a_{2,i}a_3}{\epsilon_{0,i}})|\hat\theta_i|\,.
\ea\]
From (\ref{c-i}), we can conclude that $\dot V_{\hat\theta,i}<0$ for all $|\hat\theta_i|\geq a_{0,i}+\epsilon_{0,i}$ with $i=1,\ldots,d$.
This then indicates that in the presence of dead-zone functions $\mbox{dzv}(\hat\theta)$, the trajectory $\hat\theta(t)$ of (\ref{auxisys}) is globally uniformly bounded, and will enter and stay inside the closed cube $\mathcal{B}_{0}^q$.

With the boundedness of trajectories of system (\ref{auxisys}), it thus can be deduced that there exists an $\omega$-limit set, denoted by $\omega(\mathcal{M})$, of $\mathcal{M}=\mathcal{Z}\times\mathbb{R}^d\times\mathbb{R}^q$ under the flow of (\ref{auxisys}), which is nonempty, compact and invariant, and  uniformly attracts all trajectories of (\ref{auxisys}) with initial conditions in $\mathcal{M}$.

Now we proceed to investigate the structure of this $\omega$-limit set $\omega(\mathcal{M})$. Due to the special triangular structure of (\ref{auxisys}), and by Assumption \ref{ass-MP} and the fact that the $\xi$-subsystem is globally exponentially stable at the origin, it immediately follows that on the points of $\omega(\mathcal{M})$, necessarily $\mathbf{z}\in\mathcal{Z}_c$ and $\xi=0$. As a consequence, on the $\omega$-limit set $\omega(\mathcal{M})$, $\bar\varsigma(\mathbf{z})=0$ and $\nu(\mathbf{z})=0$. In view of the previous analysis, to specify the structure of $\omega(\mathcal{M})$, we still need to determine the value of $\hat\theta$. On the other hand, when proving the boundness of $\hat\theta(t)$, we have shown that $\hat\theta(t)$ will enter and stay inside the closed cube $\mathcal{B}_{0}^q$. Thus, by recalling that $\mathcal{Z}_c$ is invariant under (\ref{sys-czd}), the value of $\hat\theta$ on $\omega(\mathcal{M})$ is determined by the properties of the system
\beeq{\label{ideal-auxzd}\ba{l}
\dot \mathbf{z} = \mathbf{f}(\mathbf{z})\,\\
\dot {\hat\theta} = \beta(\tau(\mathbf{z}))[\phi(\hat\theta,\tau(\mathbf{z}))-\phi(\theta,\tau(\mathbf{z}))]  - \mbox{dzv}(\hat\theta)\,\\
\ea}
where the initial condition $\mathbf{z}_0\in\mathcal{Z}_c$ and $\hat\theta_0\in\mathcal{B}_0^q$. It is noted that $\hat\theta(t)\in\mathcal{B}_0^q$ for all $t\geq0$ under (\ref{ideal-auxzd}).

Then, choose $V_{\tilde\theta}=\frac{1}{2}|\tilde\theta|^2$ with $\tilde\theta=\hat\theta-\theta$, whose time derivative along  (\ref{ideal-auxzd}) can be given by
\[
\dot V_{\tilde\theta} = (\hat\theta-\theta)^{\top}\beta(\tau(\mathbf{z}))[\phi(\hat\theta,\tau(\mathbf{z}))-\phi(\theta,\tau(\mathbf{z}))]-\tilde\theta^{\top}\mbox{dzv}(\hat\theta)\,.
\]
Bearing in mind the definition of $\mbox{dzv}(\cdot)$, observe that
\beeq{
(\hat\theta-\theta(\rho))^{\top}\mbox{dzv}(\hat\theta)\geq 0 \quad \mbox{for all $\hat\theta\in\mathbb{R}^p$ and $\rho\in\mathcal{P}$}\,.
}
This, together with the first part of Assumption \ref{ass-PE}, implies that under the flow (\ref{ideal-auxzd}),
\beeq{
\dot V_{\tilde\theta} \leq 0\,,
}
where the equality holds if and only if
\[\ba{l}
(\hat\theta-\theta)^{\top}\beta(\tau(\mathbf{z}))[\phi(\hat\theta,\tau(\mathbf{z}))-\phi(\theta,\tau(\mathbf{z}))]=0\,\\
(\hat\theta-\theta)^{\top}\mbox{dzv}(\hat\theta)=0\,.
\ea\]
Thus, $\hat\theta(t)$ converges to some constant value $\hat\theta^{\infty}$ as $t$ goes to infinity. By LaSalle's invariance theorem, this $\hat\theta^{\infty}$ necessarily is such that
\beeq{\label{nece-theta}\ba{l}
(\hat\theta^{\infty}-\theta)^{\top}\beta(\tau(\mathbf{z}))[\phi(\hat\theta^{\infty},\tau(\mathbf{z}))-\phi(\theta,\tau(\mathbf{z}))]=0\,\\
(\hat\theta^{\infty}-\theta)^{\top}\mbox{dzv}(\hat\theta^{\infty}) = 0\,\\
\beta(\tau(\mathbf{z}))[\phi(\hat\theta^{\infty},\tau(\mathbf{z}))-\phi(\theta,\tau(\mathbf{z}))]-\mbox{dzv}(\hat\theta^{\infty}) = 0\,.\\
\ea}
It is noted that the second of (\ref{nece-theta}) indicates that $\mbox{dzv}(\hat\theta^{\infty}) = 0$.
This further reduces (\ref{nece-theta}) to
\[
\beta(\tau(\mathbf{z}))[\phi(\hat\theta^{\infty},\tau(\mathbf{z}))-\phi(\theta,\tau(\mathbf{z}))]=0\,.
\]
By Assumption  \ref{ass-PE}.(ii), we have $\hat\theta^{\infty}=\theta$. This completes the proof. $\blacksquare$

\subsection{Proof of Lemma 2}
\label{subsec-proof-2}

Due to the special cascaded-structure of system (\ref{auxisys}) and since functions $\beta$ and $\phi$ are constructed to be globally Lipschitz and bounded, with the choice of $\lambda>0$ and Assumption \ref{ass-LES}, it is clear that the proof is completed if for any $\mathbf{z}_0\in\mathcal{Z}_c$, the origin of the \emph{linear time-varying system}
\beeq{\ba{rcl}\label{linear-sys-2}
\dst\dot {\tilde\theta} &=& \beta(\tau(\mathbf{z}(t,\mathbf{z}_0)))\dst\frac{\partial \phi(\theta,\tau(\mathbf{z}(t,\mathbf{z}_0)))}{\partial \theta}\tilde\theta
\ea}
with $\tilde\theta=\hat\theta-\theta$,
is shown to be uniformly exponentially stable.

Since $\mathbf{z}(t,\mathbf{z}_0)$ is the solution of the autonomous system (\ref{sys-czd}) passing through $\mathbf{z}_0$ at $t=0$, (\ref{linear-sys-2}) can be rewritten as a cascaded \emph{autonomous system}, having the form
\beeq{\label{linear-sys}\ba{rcl}
\dot \mathbf{z} &=& \mathbf{f}(\mathbf{z})\,\\
\dst\dot {\tilde\theta} &=& \beta(\tau(\mathbf{z}))\dst\frac{\partial \phi(\theta,\tau(\mathbf{z}))}{\partial \theta}\tilde\theta\,.\\
\ea}
We then calculate the derivative of $V_{\tilde\theta}$ as
\[
\dot V_{\tilde\theta} = \tilde\theta^{\top}\beta(\tau(\mathbf{z}))\dst\frac{\partial \phi(\theta,\tau(\mathbf{z}))}{\partial \theta}\tilde\theta\leq0
\]
where the inequality is obtained by using Assumption \ref{ass-PE}.(i).
Then, similar to the proof of Lemma \ref{lemma--1}, by LaSalle's invariance theorem and Assumption \ref{ass-PE}.(ii), we can conclude that system (\ref{linear-sys}) is uniformly asymptotically stable at the set $\mathcal{Z}_{c}\times\{0\}$, for any initial condition  $(\mathbf{z}_0,\tilde\theta_0)\in\mathcal{Z}_c\times\mathcal{R}^q$.
In other words, for any $\varepsilon>0$ and $(\mathbf{z}_0,\tilde\theta_0)\in\mathcal{Z}_c\times\mathcal{R}^q$, there exists $T_\varepsilon>0$ such that
\beeq{\label{ineq}
|\tilde\theta(t)|=\mbox{dist}\left((\mathbf{z}(t),\tilde\theta(t)),\mathcal{Z}_c\times\{0\}\right)\leq \varepsilon\quad \mbox{for all } t\geq T_\varepsilon\,.
}
Therefore, the zero equilibrium of the linear time-varying system (\ref{linear-sys-2}) is uniformly asymptotically stable, which also indicates the desired exponential stability.
 $\blacksquare$

\end{document}